\documentclass[11pt]{article}

\usepackage[utf8]{inputenc}
\usepackage[margin=1in]{geometry}
\usepackage{amsmath}
\usepackage{amssymb}
\usepackage{amsthm}
\usepackage{mathtools}
\usepackage{mathrsfs}
\usepackage{bbm}
\usepackage{bm}
\usepackage{color}
\usepackage{graphicx}
\usepackage{array}
\usepackage{multirow}
\usepackage{enumitem}
\usepackage[ruled]{algorithm2e}
\usepackage{relsize}
\usepackage[dvipsnames]{xcolor}
\usepackage[linktocpage=true,
	pagebackref=true,colorlinks,
	linkcolor=BrickRed,citecolor=blue]
	{hyperref}
\usepackage{cleveref}
\usepackage[nottoc]{tocbibind}

\usepackage{thmtools,thm-restate}
\usepackage{nicefrac}
\usepackage{makecell}

\usepackage{color-edits}
\addauthor{rdk}{blue}

\newcommand{\vb}{\boldsymbol}

\DeclareMathOperator*{\E}{\mathbb{E}}

\DeclareMathOperator{\given}{\:\vert\:}

\DeclareMathOperator*{\I}{\mathbb{I}}

\DeclareMathOperator*{\argmax}{argmax}

\DeclareMathOperator{\OPT}{OPT}
\DeclareMathOperator{\ALG}{ALG}

\DeclareMathOperator{\Ber}{Ber}
\DeclareMathOperator{\Binom}{Binom}

\newcommand{\calF}{\mathcal{F}}

\newcommand{\calI}{\mathcal{I}}

\newcommand{\calM}{\mathcal{M}}

\newcommand{\calU}{\mathcal{U}}

\let\abs\relax
\DeclarePairedDelimiter{\abs}{\lvert}{\rvert}

\newcommand{\IGNORE}[1]{}

\newcounter{note}[section]

\newcommand{\PreserveBackslash}[1]{\let\temp=\\#1\let\\=\temp}
\newcolumntype{C}[1]{>{\PreserveBackslash\centering}p{#1}}
\newcolumntype{R}[1]{>{\PreserveBackslash\raggedleft}p{#1}}
\newcolumntype{L}[1]{>{\PreserveBackslash\raggedright}p{#1}}

\newcommand{\ols}[1]{\mskip.5\thinmuskip\overline{\mskip-.5\thinmuskip {#1} \mskip-.5\thinmuskip}\mskip.5\thinmuskip} % overline short
\newcommand{\olsi}[1]{\,\overline{\!{#1}}} % overline short italic
\makeatletter
\newcommand\closure[1]{
  \tctestifnum{\count@stringtoks{#1}>1} %checks if number of chars in arg > 1 (including '\')
  {\ols{#1}} %if arg is longer than just one char, e.g. \mathbb{Q}, \mathbb{F},...
  {\olsi{#1}} %if arg is just one char, e.g. K, L,...
}
\long\def\count@stringtoks#1{\tc@earg\count@toks{\string#1}}
\long\def\count@toks#1{\the\numexpr-1\count@@toks#1.\tc@endcnt}
\long\def\count@@toks#1#2\tc@endcnt{+1\tc@ifempty{#2}{\relax}{\count@@toks#2\tc@endcnt}}
\def\tc@ifempty#1{\tc@testxifx{\expandafter\relax\detokenize{#1}\relax}}
\long\def\tc@earg#1#2{\expandafter#1\expandafter{#2}}
\long\def\tctestifnum#1{\tctestifcon{\ifnum#1\relax}}
\long\def\tctestifcon#1{#1\expandafter\tc@exfirst\else\expandafter\tc@exsecond\fi}
\long\def\tc@testxifx{\tc@earg\tctestifx}
\long\def\tctestifx#1{\tctestifcon{\ifx#1}}
\long\def\tc@exfirst#1#2{#1}
\long\def\tc@exsecond#1#2{#2}
\makeatother

\newtheorem{theorem}{Theorem}[section]

\Crefname{claim}{Claim}{Claims}
\newtheorem{proposition}[theorem]{Proposition}
\newtheorem{lemma}[theorem]{Lemma}
\newtheorem{corollary}[theorem]{Corollary}

\newtheorem{fact}[theorem]{Fact}
\theoremstyle{definition}

\newtheorem{definition}[theorem]{Definition}

\newtheorem{remark}[theorem]{Remark}

\title{The $k$-Fold Matroid Secretary Problem}
\author{Rishi Gujjar \and Kevin Hua \and Robert Kleinberg \and Frederick V. Qiu}
\date{\today}

\begin{document}

\maketitle

\begin{abstract}
    In the \emph{matroid secretary problem}, elements $N \coloneqq [n]$ of a matroid $\calM \subseteq 2^N$ arrive in random order. When an element arrives, its weight is revealed and a choice must be made to accept or reject the element, subject to the constraint that the accepted set $S \in \calM$.~\cite{Kleinberg05} gives a $(1-O(1/\sqrt{k}))$-competitive algorithm when $\calM$ is a $k$-uniform matroid. We generalize their result, giving a $(1-O(\sqrt{\log(n)/k}))$-competitive algorithm when $\calM$ is a \emph{$k$-fold matroid union}.
\end{abstract}

\section{Introduction}

In the secretary problem, elements $N \coloneqq [n]$ arrive in random order. When an element $i \in N$ arrives, its weight is revealed, and a decision must be made to accept or reject the element, subject to a known feasibility constraint $\calF$. The objective is to maximize the expected value of the accepted set. If an algorithm accepts a set whose expected weight is at least $\alpha$ times the weight of the max-weight set in $\calF$, we say the algorithm is \emph{$\alpha$-competitive}.

Secretary problems are fundamental problems in optimal stopping with close connections to economics and computing, relating to pricing and auction theory~\cite{HajiaghayiKP04,Kleinberg05,BabaioffIK07,BabaioffIKK08}. The classical result of~\cite{Dynkin63} gives an optimal \emph{$1/e$-competitive} algorithm when $\calF$ is a $1$-uniform matroid, i.e., at most $1$ element can be selected. The later work of~\cite{Kleinberg05} gives a $(1-5/\sqrt{k})$-competitive algorithm when $\calF$ is a $k$-uniform matroid, i.e., at most $k$ elements can be selected.\footnote{This result is optimal in the sense that there exists a constant $C$ for which no algorithm is $(1-C/\sqrt{k})$-competitive.} One might ask what conditions on $\calF$ allow for $(1-o(1))$-competitive algorithms. Motivated by the recent work of~\cite{AlonGPRWWZ25} for prophet inequalities, we consider the following class of constraints.

\begin{definition}[Matroid Union]
    For matroids $\calM_1, \dots, \calM_k \subseteq 2^N$, their \emph{union} $\calM_1 \lor \dots \lor \calM_k$ is the matroid where $S \in \calM_1 \lor \dots \lor \calM_k$ if and only if $S$ can be partitioned into $S_1, \dots, S_k$ such that $S_\ell \in \calM_\ell$ for all $\ell \in [k]$. For a matroid $\calM$, its \emph{$k$-fold union} is defined as $\calM^k = \bigvee_{i \in [k]} \calM$.
\end{definition}

Observe that a $k$-uniform matroid is the $k$-fold union of a $1$-uniform matroid, and so the class of $k$-fold matroid unions nicely generalizes a ``redundancy'' property that allows for $(1-o(1))$-competitive algorithms in the $k$-uniform matroid case. Our main result is as follows.

\begin{theorem} \label{thm:kFoldMatroidSecretary}
    For any matroid $\calM$, there is a $(1-O(\sqrt{\log(n)/k}))$-competitive algorithm for $\calM^k$.
\end{theorem}

The weaker statement that the competitive ratio is $1-O(\sqrt{\log(n)}/k^{1/3})$ has a straightforward proof when $\calM^k$ is a $k$-uniform matroid: let $\varepsilon = \Theta(\sqrt{\log(n)}/k^{1/3})$, and accept every element with higher weight than the $\approx \varepsilon k$\textsuperscript{th} highest weight element among the first $\varepsilon n$ elements that arrive. By standard concentration inequalities, the number of elements exceeding this threshold is in $[(1-O(\varepsilon)) k, k]$ w.h.p. Therefore, accepting every element that exceeds the threshold yields a $(1-O(\sqrt{\log(n)}/k^{1/3}))$-competitive algorithm. To generalize to $k$-fold matroid unions, we no longer consider the number of accepted elements, but the \emph{covering number} of the accepted set.

\begin{definition}[Covering Number]
    The \emph{covering number $\varphi(S, \calM)$ of $S$ in $\calM$} is defined to be the smallest $r \in \mathbb{N}$ such that $S$ can be partitioned into $S_1, \dots, S_r \in \calM$.
\end{definition}

Observe that if $\varphi(S, \calM) \leq k$, then $S \in \calM^k$, so $\varphi(S, \calM)$ behaves as an analogue of $\abs{S}$ when generalizing from $k$-uniform matroids to $k$-fold matroid unions. The difficulty lies in the fact that while the cardinality function is additive and thus satisfies many black-box properties regarding its expectation and concentration, it is not a priori clear how to control the expectation or concentration of the covering number $\varphi$.

Instead, our analysis relies on a characterization of the covering number via a generalization of the Nash-Williams Theorem~\cite{Tutte61,Nash-Williams64,Edmonds65}. A slight manipulation of the theorem statement allows us to bound the covering number by a maximum over additive functions, whose upper tail can be controlled with a union bound. We then use the upper tail bound to argue that our algorithm is able to accept every element that exceeds its threshold w.h.p. Some additional care is needed to argue that the expected weight of the accepted elements is large.

One can ask whether $(1-o(1))$-competitive algorithms exist under broader feasibility constraints than $k$-fold matroid unions. One natural weakening is the class of matroids resulting from unions of $k$ possibly different loopless matroids. Note that every loopless matroid contains the $1$-uniform matroid, so every union of $k$ loopless matroids contains the $k$-uniform matroid. We give a simple example showing that such a condition does not suffice to obtain a $1 - o(1)$ competitive ratio.

\begin{theorem}[name=,restate=kMatroids] \label{thm:kMatroids}
    For any $\varepsilon > 0$ and $k \in \mathbb{N}$, there exists $n \in \mathbb{N}$ and a loopless matroid $\calM \subseteq 2^N$ such that no algorithm is $(1/e + \varepsilon)$-competitive for $\calM \lor \calU^k$, where $\calU \subseteq 2^N$ is the $1$-uniform matroid.
\end{theorem}

The construction relies on the fact that the rank of $\calM$ can be much larger than $k$, so allowing an algorithm to select an additional $k$ elements has little effect.

\subsection{Related Work}

We have already discussed the most relevant works to ours:~\cite{Kleinberg05} considers the special case of $k$-uniform matroids, and~\cite{AlonGPRWWZ25} considers $k$-fold matroid unions in the related prophet inequality setting. Prior to their work,~\cite{ChekuriSZ24} studied $k$-fold matroid unions in offline optimization settings, and asked whether similar results can be obtained in online settings. For hardness results,~\cite{BanihashemHKKMO25} constructs high girth \emph{graphic} matroids for which no algorithm is $(1/e+\Omega(1))$-competitive.

Besides these works, there is a long line of work in the secretary problem when the feasibility constraint is a matroid. The problem was introduced in the seminal work of~\cite{BabaioffIK07,BabaioffIKK18}, which posed the \emph{(strong) matroid secretary conjecture}, that a constant-competitive ($1/e$-competitive) algorithm exists for the matroid secretary problem. Both conjectures remain unresolved: the best known algorithms for general matroids are $\Omega(1/\log\log(n))$-competitive~\cite{Lachish14,FeldmanSZ18} and there are no hardness results beyond those for $1$-uniform matroids. There are positive results for special classes of matroids~\cite{KesselheimRTV13,DinitzK14,BercziLSV25,BanihashemHKKMO25}, as well as hard instance constructions which rule out certain classes of algorithms~\cite{BahraniBSW21,AbdolazimiKKG23}.

\section{Preliminaries}

\paragraph{Secretary Problem.}
In the secretary problem, we have a set of elements $N \coloneqq [n]$ with distinct non-negative weights and a feasibility constraint $\calF \subseteq 2^N$. The feasibility constraint is known in advance, and the weights of the elements are unknown. In our setting, $\calF = \calM^k$ is assumed to be a $k$-fold matroid union. The elements arrive one at a time in uniformly random order. When element $i \in N$ arrives, its weight is revealed, and the algorithm must decide whether to accept or reject the element, subject to only accepting the element when $\ALG \cup \{i\} \in \calF$, where $\ALG$ is the current set of accepted elements.

We denote by $w(S)$ the weight of a set $S \subseteq N$, which is equal to the sum of the weights of the elements of $S$. We say that an algorithm accepting elements $\ALG$ is \emph{$\alpha$-competitive} if $\E[w(\ALG)] \geq (1/\alpha) \cdot \max_{S \in \calF} w(S)$.

\paragraph{Matroids.}
We now define some terminology, notation, and properties of matroids, skipping over anything not strictly necessary for our technical results.

\begin{definition}[Matroid]
    A set family $\calF \subseteq 2^N$ is a \emph{matroid} if the following properties hold:
    \begin{enumerate}[topsep=4pt]
        \item $\emptyset \in \calF$.
        
        \item If $S \subseteq T$ and $T \in \calF$, then $S \in \calF$ (downward-closure).

        \item If $S, T \in \calF$ and $\abs{S} < \abs{T}$, there exists $i \in T \setminus S$ such that $S \cup \{i\} \in \calF$ (augmentation).
    \end{enumerate}
\end{definition}

For the remaining definitions, let $\calM \subseteq 2^N$ be a matroid.\footnote{A matroid is sometimes defined as a tuple $\calM = (N, \calI)$ of the ground set and the sets contained in the matroid, but for our purposes, it is more convenient to let $\calM$ be synonymous with the set-family itself.}

\begin{definition}[Rank]
    The \emph{rank} of a set $S$ w.r.t.\ $\calM$ is defined $\rho(S, \calM) = \max_{T \in \calM} \abs{S \cap T}$.
\end{definition}

\begin{definition}[Basis]
    A set $S \in \calM$ is a \emph{basis} if $\abs{S} = \rho(S, \calM) = \rho(N, \calM)$.
\end{definition}

\begin{definition}[Max-Weight Basis]
    The \emph{max-weight basis} of $S$ w.r.t.\ $\calM$ is defined $\OPT(S, \calM) = \argmax_{T \subseteq S, T \in \calM} w(T)$. Note that $\OPT(S, \calM)$ is unique when the element weights are unique.
\end{definition}

\begin{definition}[Span]
    The \emph{span} of a set $S$ w.r.t.\ $\calM$ is defined $\sigma(S, \calM) = \{i \in N : \rho(S \cup \{i\}, \calM) = \rho(S, \calM)\}$.
\end{definition}

\begin{definition}[Flat]
    A set $S$ is a \emph{flat} of $\calM$ if $S = \sigma(S, \calM)$.
\end{definition}

\begin{definition}[Improves]
    An element $i$ \emph{improves} a set $S$ w.r.t.\ $\calM$ if $i \in \OPT(S \cup \{i\}, \calM)$. Note that every element $i \in \OPT(S, \calM)$ trivially improves $S$ w.r.t.\ $\calM$. 
\end{definition}

\begin{fact} \label{fact:Improves}
    $i$ improves $S$ w.r.t.\ $\calM$ if and only if $i \not\in \sigma(\{j \in S : w(\{j\}) > w(\{i\})\}, \calM)$.
\end{fact}

Finally, we present the characterization of the covering number we will use in our analysis.

\begin{theorem}[\cite{Tutte61,Nash-Williams64,Edmonds65}]
    For all $S \subseteq N$,
    \[
        \varphi(S, \calM) \quad = \quad \max_{T \subseteq S} \bigg\lceil \frac{\abs{T}}{\rho(T, \calM)} \bigg\rceil \enspace .
    \]
\end{theorem}

\begin{corollary} \label{cor:NashWilliams}
    Let $\calF \subseteq 2^N$ be the collection of flats of $\calM$. Then for all $S \subseteq N$,
    \[
        \varphi(S, \calM) \quad = \quad \max_{F \in \calF} \bigg\lceil \frac{\abs{S \cap F}}{\rho(F, \calM)} \bigg\rceil \enspace .
    \]
\end{corollary}
\begin{proof}
Observe that
\begin{align*}
    \max_{T \subseteq S} \bigg\lceil \frac{\abs{T}}{\rho(T, \calM)} \bigg\rceil \quad 
    &= \quad \max_{T \subseteq S} \bigg\lceil \frac{\abs{S \cap T}}{\rho(T, \calM)} \bigg\rceil \quad = \quad \max_{T \subseteq N} \bigg\lceil \frac{\abs{S \cap T}}{\rho(T, \calM)} \bigg\rceil \quad = \quad \max_{F \in \calF} \bigg\lceil \frac{\abs{S \cap F}}{\rho(F, \calM)} \bigg\rceil \enspace .
\end{align*}
The second equality holds because maximizing over a larger set can only increase the maximum, but all the newly considered $T \not\subseteq S$ can only decrease the value of $\abs{S \cap T}/\rho(T, \calM)$. The third equality holds because if the maximum is attained at $T \subseteq N$, it is also attained at the flat $\sigma(T, \calM)$, as taking the span can only increase the numerator while keeping the denominator the same.
\end{proof}

\section{Algorithm}

Our algorithm takes as input a matroid $\calM \subseteq 2^N$ and a positive integer $k$ representing the feasibility constraint $\calM^k$. In this section, ranks and covering numbers are always with respect to the matroid $\calM$, so we directly notate, e.g., $\varphi(S) \coloneqq \varphi(S, \calM)$.

Suppose that the elements are sorted in decreasing order by weight, i.e., $w(\{1\}) > \dots > w(\{n\})$, as the behavior of our algorithm is invariant under labeling of the elements.

Let $k \leq n$, or else accepting every possible element yields the max-weight basis. Also let $k = \Omega(\log(n))$ be sufficiently large, or else any algorithm is vacuously $(1 - O(\sqrt{\log(n)/k}))$-competitive.

Define $\varepsilon : \mathbb{R}_+ \to \mathbb{R}_+$ by $\varepsilon(x) = C\sqrt{\log(n)/x}$, where $C$ is a constant chosen such that it is sufficiently large and satisfies $L \coloneqq \log(1/\varepsilon(k)) - 2$ is a positive integer. For example, we can let $k \geq 160^2\log(n)$ and $C \in [10, 20]$.

Define $(Z_0, Z_1, \dots, Z_L) \sim \mu$ as follows: throw $n$ balls into bins $0, 1, \dots, L$ so that each ball independently lands in bin $\ell$ w.p.\ $2^{\max\{1, \ell\}-L-1}$. Then each $Z_\ell$ is the number of balls in bin $\ell$.

\paragraph{Explanation of the Algorithm.}
Our algorithm is presented as \Cref{alg:kFoldMatroidSecretary} below. It consists of a sequence of $L$ phases, with each phase doubling in length and using the previously observed elements to choose a thresholding rule for the newly observed arrivals. Each phase is allowed to accept elements so long as their covering number is not much larger proportionally than the length of the phase. For example, phase $L$, which processes $\approx n/2$ elements, should only accept a subset of elements $\ALG_L$ so that $\varphi(\ALG_L) \approx k/2$. We now describe how the algorithm behaves in the final phase, which captures its behavior in every phase.

In the final phase $L$, the algorithm has so far observed the first $Z_0 + Z_1 + \dots + Z_{L-1}$ arrivals, which constitute half the elements in expectation. The other $Z_L$ elements have yet to arrive. The algorithm can then be thought of as sub-sampling each element w.p.\ $1-\varepsilon(r_L)$, denoting by $A_{\to L}$ the sub-sampled elements of the already observed half, and denoting by $A_L$ the sub-sampled elements of the remaining half. We can think of $A_{\to L}$ as a ``sample'' which determines the thresholding rule for the phase, and of $A_L$ as the new arrivals ``eligible'' for selection. The explicit sub-sampling allows the concentration arguments to go through cleanly in the analysis. Note that it is not strictly necessary to define $A_L$, but the notation makes the analysis clearer.

The algorithm then processes the $Z_L$ remaining elements, accepting an element if improves $A_{\to L}$, is in $A_L$, and does not violate the phase's capacity constraint.

\begin{algorithm}[ht]
    \caption{Algorithm for the secretary problem with feasibility constraint $\calM^k$.}
    \label{alg:kFoldMatroidSecretary}
    
    $\ALG \gets \emptyset$\;

    $(Z_0, Z_1, \dots, Z_L) \sim \mu$\;
    
    Observe the first $Z_0$ arrivals\;

    \For{\emph{phase} $\ell \in [L]$}{
        $\ALG_\ell, A_{\to\ell}, A_\ell \gets \emptyset$\;
        
        $r_\ell \gets 2^{\ell-L-1} k$\;

        \For{$i \in \text{\emph{all observed arrivals}}$}{
            $X \sim \Ber(1-\varepsilon(r_\ell))$\;

            \If{$X = 1$}{
                $A_{\to\ell} \gets A_{\to\ell} \cup \{i\}$\;
            }
        }
        
        \For{$i \in \text{\emph{the next $Z_\ell$ arrivals}}$}{
            $X \sim \Ber(1-\varepsilon(r_\ell))$\;

            \If{$X = 1$}{
                $A_{\ell} \gets A_{\ell} \cup \{i\}$\;

                \If{$i$ \emph{improves} $A_{\to\ell}$ \emph{w.r.t.}\ $\calM^{r_\ell}$ \emph{and} $i \not\in \sigma(\ALG_\ell, \calM^{(1 + \varepsilon(r_\ell)) r_\ell})$}{
                    $\ALG_\ell \gets \ALG_\ell \cup \{i\}$\;
                }
            }
            
        }

        $\ALG \gets \ALG \cup \ALG_\ell$\;
    }

    \Return $\ALG$\;
\end{algorithm}

\paragraph{Analysis.}
The analysis defines a few sets of interest: $S^+ \subseteq A_{\to\ell}$ is the overall max-weight basis intersected with the sample, $S^* \subseteq A_{\to\ell}$ is the max-weight basis of the sample w.r.t.\ the scaled down feasibility constraint, $T^* \subseteq A_\ell$ is the set of eligible new arrivals which improve the sample, and $\ALG_\ell$ is the subset of $T^*$ accepted by the algorithm given the scaled covering number constraint. Note that $\E[w(S^+)]$ is the actual benchmark we want to compare $\E[w(\ALG_\ell)]$ against.

We argue by concentration that $T^* = \ALG_\ell$ w.h.p.\ (\Cref{lemma:ImprovementIsIndependent}), and that $\E[w(S^+)] \approx \E[w(S^*)]$ (\Cref{lemma:ThinningAlsoThinsOccupancy}). Though not immediately apparent, we also have $\E[w(S^*)] = \E[w(T^*)]$ by symmetry, since both $S^*$ and $T^*$ are the elements which improve $A_{\to\ell}$ (\Cref{lemma:ExpectedUtility}). Putting everything together yields $\E[w(S^+)] \approx \E[w(S^*)] = \E[w(T^*)] \approx \E[w(\ALG_\ell)]$ (\Cref{prop:ExpectedUtility}).

We now proceed with the formal statements and proofs.

\begin{proposition} \label{prop:ExpectedUtility}
    Let $p \in [\varepsilon(k), 1/2]$ and $r \geq (1+\varepsilon(pk))pk$. Let $X_1, \dots, X_n \sim \Ber(2p)$ and $Y_1, \dots, Y_n \sim \Ber(1/2)$ be mutually independent, and define $S = \{i \in N : X_i = 1, Y_i = 1\}$ and $T = \{i \in N : X_i = 1, Y_i = 0\}$. Let $T^* \subseteq T$ be the set of elements that improve $S$ w.r.t.\ $\calM^r$. Then
    \[
        \E\Big[w(T^*) \cdot \I\Big(\varphi(T^*) < (1+\varepsilon(r)) r\Big)\Big] \quad \geq \quad p(1 - 2/n) \cdot w(\OPT(N, \calM^k)) \enspace .
    \]
\end{proposition}

We first prove some helpful lemmas, which reuse the notation defined in~\Cref{prop:ExpectedUtility}. For all $j \in [\rho(N)]$, let $\calF_j \subseteq \calM$ be the collection of flats of rank $j$ in $\calM$, and observe that $\abs{\calF_j} \leq n^j$. Let $S^* = \OPT(S, \calM^r)$ and $S^+ = S \cap \OPT(N, \calM^k)$.

\begin{lemma} \label{lemma:ImprovementIsIndependent}
    \[
        \Pr[\varphi(T^*) \geq (1+\varepsilon(r)) r] \quad \leq \quad 1/n^3 \enspace .
    \]
\end{lemma}
\begin{proof}
Fix some $j \in [\rho(N)]$ and some $F \in \calF_j$. Initialize $U = \emptyset$, and for each $i = 1, \dots, n$ in sequence, set $U = U \cup \{i\}$ if $i \in F$, $X_i = 1$, and $i \not\in \sigma(S \cap [i-1], \calM^r)$. Since $w(\{1\}) > \dots > w(\{n\})$, we have by \Cref{fact:Improves} that $U$ is the set of elements of $(S \cup T) \cap F$ which improve $S$ w.r.t.\ $\calM^r$. Hence, $\abs{S^* \cap F} = \abs{\{i\in U: Y_i = 1\}}$ and $\abs{T^* \cap F} = \abs{\{i\in U: Y_i = 0\}}$.

Observe that $\{Y_i : i \in U\}$ is a sequence of independent coin flips. Additionally, $\varphi(S^*) \leq r$ because $S^* = \OPT(S, \calM^r)$. Therefore, we have $\abs{S^* \cap F} \leq rj$ by~\Cref{cor:NashWilliams}. Then the probability that $\abs{T^* \cap F} \geq (1+\varepsilon(r)) rj$ is equal to the probability that there are at most $rj$ heads among at least $(2+\varepsilon(r)) rj$ coin flips. Thus, $\Pr[\abs{T^* \cap F} \geq (1 + \varepsilon(r)) rj] \leq 1/n^{5j}$ by a Chernoff bound. Taking a union bound over all $F \in \calF_j$ and $j \in [\rho(N)]$, then applying~\Cref{cor:NashWilliams} completes the proof.
\end{proof}

\begin{lemma} \label{lemma:ThinningAlsoThinsOccupancy}
    \[
        \Pr[\varphi(S^+) \geq (1 + \varepsilon(pk))pk] \quad \leq \quad 1/n^3 \enspace .
    \]
\end{lemma}
\begin{proof}
Fix some $j \in [\rho(N)]$ and some $F \in \calF_j$. By~\Cref{cor:NashWilliams}, $\abs{\OPT(N, \calM^k) \cap F} \leq kj$. Since $S^+$ independently includes each element of $\OPT(N, \calM^k)$ with probability $p$, $\abs{S^+ \cap F}$ is stochastically dominated by $X \sim \Binom(kj, p)$. Therefore, the probability that $\abs{S^+ \cap F} \geq (1 + \varepsilon(pk))pkj$ is at most $1/n^{5j}$ by a Chernoff bound. Taking a union bound over all $F \in \calF_j$ and $j \in [\rho(N)]$, then applying~\Cref{cor:NashWilliams} completes the proof.
\end{proof}

\begin{lemma} \label{lemma:ExpectedUtility}
    \[
        \E[w(T^*)] \quad = \quad \E[w(S^*)] \enspace .
    \]
\end{lemma}
\begin{proof}
Observe that
\begin{align*}
    \E[w(T^*)] \quad &= \quad \sum_{i=1}^n \Big(\Pr[i \in T^*] \cdot w(\{i\})\Big) \\
    &= \quad \sum_{i=1}^n \Big(\Pr[i \in T^* \given i \in T] \Pr[i \in T] \cdot w(\{i\})\Big) \\
    &= \quad p \sum_{i=1}^n \Big(\Pr[i \not\in \sigma(S \cap [i-1], \calM^r) \given i \in T] \cdot w(\{i\})\Big) \\
    &= \quad p \sum_{i=1}^n \Big(\Pr[i \not\in \sigma(S \cap [i-1], \calM^r)] \cdot w(\{i\})\Big) \enspace ,
\end{align*}
where the third equality follows by definition of $T^*$ and~\Cref{fact:Improves}, and the last equality follows because $i \in \sigma(S \cap [i-1], \calM^r)$ depends only on elements $[i-1]$ and is therefore independent from $i \in T$. Repeating the argument for $\E[w(S^*)]$ yields the same conclusion, completing the proof.
\end{proof}

\begin{proof}[Proof of~\Cref{prop:ExpectedUtility}]
Let $i$ be the highest weight element such that $\{i\} \in \calM$. Observe that $\Pr[i \in T^*] = \Pr[i \in T] = p \geq 1/n$. Then by~\Cref{lemma:ImprovementIsIndependent},
\begin{align*}
    \E[w(T^*)] \quad &= \quad \E[w(T^*) \cdot \I(\varphi(T^*) < (1+\varepsilon(r))r)] + \E[w(T^*) \cdot \I(\varphi(T^*) \geq (1+\varepsilon(r))r)] \\
    &\leq \quad \E[w(T^*) \cdot \I(\varphi(T^*) < (1+\varepsilon(r))r)] + n \cdot w(\{i\}) \cdot \Pr[\varphi(T^*) \geq (1+\varepsilon(r))r] \\
    &\leq \quad \E[w(T^*) \cdot \I(\varphi(T^*) < (1+\varepsilon(r))r)] + n^2 \E[w(T^*)] \Pr[\varphi(T^*) \geq (1+\varepsilon(r))r] \\
    &\leq \quad \E[w(T^*) \cdot \I(\varphi(T^*) < (1+\varepsilon(r))r)] + \E[w(T^*)]/n \enspace .
\end{align*}

Similarly, $\Pr[i \in S^+] = \Pr[i \in S] = p \geq 1/n$, so by~\Cref{lemma:ThinningAlsoThinsOccupancy},
\begin{align*}
    \E[w(S^+)] \quad &= \quad \E[w(S^+) \cdot \I(\varphi(S^+) < r)] + \E[w(S^+) \cdot \I(\varphi(S^+) \geq r)] \\
    &\leq \quad \E[w(S^+) \cdot \I(\varphi(S^+) < r)] + n \cdot w(\{i\}) \cdot \Pr[\I(\varphi(S^+) \geq r)] \\
    &\leq \quad \E[w(S^+) \cdot \I(\varphi(S^+) < r)] + n^2 \E[w(S^+)] \Pr[\varphi(S^+) \geq r] \\
    &\leq \quad \E[w(S^+) \cdot \I(\varphi(S^+) < r)] + \E[w(S^+)]/n \enspace .
\end{align*}

Chaining the above two inequalities and~\Cref{lemma:ExpectedUtility} together gives us
\begin{align*}
    \E[w(T^*) \cdot \I(\varphi(T^*) < (1+\varepsilon(r))r)] \quad &\geq \quad (1 - 1/n) \E[w(T^*)] \\
    &= \quad (1 - 1/n) \E[w(S^*)] \\
    &\geq \quad (1 - 1/n) \E[w(S^+) \cdot \I(\varphi(S^+) < r)] \\
    &\geq \quad (1 - 2/n) \E[w(S^+)] \\
    &= \quad p(1 - 2/n) \cdot w(\OPT(N, \calM^k)) \enspace . \qedhere
\end{align*}
\end{proof}

\begin{theorem}
    \Cref{alg:kFoldMatroidSecretary} is $(1 - O(\sqrt{\log(n)/k}))$-competitive.
\end{theorem}
\begin{proof}
Observe that $\ALG = \ALG_1 \cup \dots \cup \ALG_{\log(1/\varepsilon)} \in \calM^k$ because
\begin{align*}
    \sum_{\ell=1}^{L} (1+\varepsilon(r_\ell)) r_\ell \quad &= \quad \sum_{\ell=1}^{L} 2^{\ell-L-1} k + \sum_{\ell=1}^L C\sqrt{2^{\ell-L-1} k\log(n)} \\
    &= \quad (1 - 2^{-L}) k + C\sqrt{k\log(n)} \sum_{\ell=1}^L \sqrt{2^{\ell-L-1}} \\
    &= \quad k - 4C\sqrt{k\log(n)} + C\sqrt{k\log(n)} \frac{\sqrt{2^{-L}} - \sqrt{2}}{1 - \sqrt{2}} \quad < \quad k \enspace .
\end{align*}

Recall that $A_{\to\ell}$ is a $(1-\varepsilon(r_\ell))$-subsampling of the elements observed before phase $\ell$, and that $A_\ell$ is a $(1-\varepsilon(r_\ell))$-subsampling of the elements observed during phase $\ell$. For all $\ell \in [L]$, let $A^*_\ell \subseteq A_\ell$ be the set of elements that improve $A_{\to\ell}$ w.r.t.\ $\calM^{r_\ell}$. Observe that $A_{\to\ell}, A_\ell$ satisfy the conditions of~\Cref{prop:ExpectedUtility} with $p = 2^{\ell-L-1}(1-\varepsilon(r_\ell))$ and $r = r_\ell$. Thus,
\begin{align*}
    \E[w(\ALG)] \quad &= \quad \sum_{\ell=1}^L \E[w(\ALG_\ell)] \\
    &\geq \quad \sum_{\ell=1}^L \E[w(A^*_\ell) \cdot \I(\varphi(A^*_\ell) < (1 + \varepsilon(r_\ell)) r_\ell)] \\
    &\geq \quad (1-2/n) \cdot w(\OPT(N, \calM^k)) \sum_{\ell=1}^L 2^{\ell-L-1}(1-\varepsilon(r_\ell)) \\
    &= \quad (1-2/n) \cdot w(\OPT(N, \calM^k)) \bigg(1 - 2^{-L} - C\sqrt{\log(n)/k}\sum_{\ell=1}^L \sqrt{2^{\ell-L-1}}\bigg) \\
    &\geq \quad (1-2/n - 8\varepsilon(k)) \cdot w(\OPT(N, \calM^k)) \\
    &= \quad (1 - O(\sqrt{\log(n)/k})) \cdot w(\OPT(N, \calM^k)) \enspace . \qedhere
\end{align*}
\end{proof}

\begin{remark}
    Let $\sim$ be the equivalence relation which identifies \emph{parallel} elements ($i, j$ are parallel w.r.t.\ $\calM$ if $\rho(\{i\}) = \rho(\{j\}) = \rho(\{i, j\}) = 1$), and let $n_{\sim} = \abs{N / {\sim}}$ (the number of equivalence classes under the relation $\sim$). Since parallel elements behave identically with respect to spans, we can get the tighter bound $\abs{\calF_j} \leq n_\sim^j$. Repeating the proof with $\varepsilon(x) = \Theta(\sqrt{\log(n_\sim)/x})$ shows that~\Cref{alg:kFoldMatroidSecretary} is $(1 - O(\sqrt{\log(n_\sim)/k}))$-competitive. This is an improvement in certain cases; for example, graphic matroids satisfy $n_\sim \leq (\rho(N) + 1)^2$, so the dependence on $\log(n)$ can be replaced with a dependence on $\log(\rho(N))$.
\end{remark}

\section{Unions of \texorpdfstring{$k$}{k} Matroids}

The formal proof of~\Cref{thm:kMatroids} uses a similar approach to the proof of~\cite{BanihashemHKKMO25} for high girth graphic matroids, although our construction is much simpler because we are not constrained to the structure of graphic matroids.

\begin{definition}[Distributionally Competitive]
    Let $\mu$ be a distribution over a finite subset of $\mathbb{R}_+^m$ and $\calM \subseteq 2^{[m]}$ be a matroid. Let $\ALG$ be an algorithm, and for $\vb{v} \in \mathbb{R}_+^m$, let $\ALG(\vb{v}, \calM)$ be a random variable denoting the set of elements accepted by $\ALG$ on the secretary problem instance with element weights $\vb{v}$ and feasibility constraint $\calM$. Let $\OPT(\vb{v}, \calM)$ be the max-weight basis of $\calM$ w.r.t.\ the weights $\vb{v}$. We say $\ALG$ is \emph{$\alpha$-competitive for $\calM$ over $\mu$} if
    \[
        \E_{\vb{v} \sim \mu}[w(\ALG(\vb{v}, \calM))] \quad \geq \quad \alpha \cdot \E_{\vb{v} \sim \mu}[w(\OPT(\vb{v}, \calM))] \enspace .
    \]
\end{definition}

\begin{theorem}[\cite{BanihashemHKKMO25}] \label{thm:DistributionallyCompetitive}
    Let $\mu$ be a distribution over a finite subset of $\mathbb{R}_+^N$ and $\calM \subseteq 2^N$ be a matroid. If no algorithm is $\alpha$-competitive for $\calM$ over $\mu$, then no algorithm is $\alpha$-competitive for $\calM$.
\end{theorem}

\begin{theorem}[\cite{BanihashemHKKMO25}] \label{thm:HardDistribution}
    For any $\varepsilon > 0$, there exists $m \in \mathbb{N}$ and a distribution $\mu$ over a finite subset of $\mathbb{R}_+^m$ such that for element weights $(v_1, \dots, v_m) \sim \mu$, no algorithm is $(1/e+\varepsilon)$-competitive for $\calU_m$ over $\mu$, where $\calU_m$ is the $1$-uniform matroid on $m$ elements.
\end{theorem}

\kMatroids*
\begin{proof}
Without loss of generality, suppose $2k/\varepsilon$ is an integer. Let $m$ and $\mu$ be those promised by~\Cref{thm:HardDistribution} for $\varepsilon/2$, and let $n = 2km/\varepsilon$. Let $\calM \subseteq 2^N$ be the disjoint union of $2k/\varepsilon$ copies of $\calU_m$, i.e., $\calM = \{S \subseteq N : \forall i \in [2k/\varepsilon], \abs{S \cap [(i-1)m + 1, im]} \leq 1\}$.

Let $\nu$ be the distribution over a finite subset of $\mathbb{R}_+^n$ obtained by concatenating $2k/\varepsilon$ independent draws from $\mu$. Suppose for contradiction that there exists an algorithm $\ALG$ which is $(1/e+\varepsilon)$-competitive for $\calM \lor \calU^k$ over $\nu$. Then consider the following algorithm $\ALG'$ for the $m$-element secretary problem for $\calU_m$ over $\mu$:
\begin{itemize}[topsep=4pt]
    \item Sample $\vb{v} \sim \nu$ and uniformly random $j \in [2k/\varepsilon]$. Replace the weights of elements $[(j-1)m+1, jm]$ with the input to $\ALG'$.

    \item Run $\ALG$ on $\vb{v}$, where we can simulate the arrivals of elements in $[(j-1)m+1, jm]$ by observing the next arrival of the input to $\ALG'$.

    \item The first time $\ALG$ accepts an element in $[(j-1)m+1, jm]$, accept the corresponding element observed by $\ALG'$.
\end{itemize}

Observe that
\begin{align*}
    \E_{\vb{v} \sim \nu}[w(\ALG(\vb{v}, \calM \lor \calU^k))] \quad &\geq \quad (1/e + \varepsilon) \E_{\vb{v} \sim \nu}[w(\OPT(\vb{v}, \calM \lor \calU^k))] \\
    &\geq \quad (1/e + \varepsilon)(2k/\varepsilon) \E_{\vb{u} \sim \mu}[w(\OPT(\vb{u}, \calU_m))] \enspace .
\end{align*}

Let $\ALG_\calM(\vb{v}, \calM \lor \calU^k)$ be the set of first elements accepted from each $[(i-1)m+1, im]$. Since this set excludes at most $k$ elements from $\ALG(\vb{v}, \calM \lor \calU^k)$,
\begin{align*}
    \E_{\vb{v} \sim \nu}[w(\ALG_\calM(\vb{v}, \calM \lor \calU^k))] \quad &\geq \quad \E_{\vb{v} \sim \nu}[w(\ALG(\vb{v}, \calM \lor \calU^k))] - k \cdot \E_{\vb{u} \sim \mu}[w(\OPT(\vb{u}, \calU_m))] \\
    &\geq \quad (1/e + \varepsilon/2)(2k/\varepsilon) \E_{\vb{u} \sim \mu}[w(\OPT(\vb{u}, \calU_m))] \enspace .
\end{align*}

Finally, because $j \in [2k/\varepsilon]$ is uniformly random,
\begin{align*}
    \E_{\vb{u} \sim \mu}[w(\ALG'(\vb{u}, \calU_m))] \quad &= \quad \varepsilon/(2k) \E_{\vb{v} \sim \nu}[w(\ALG_\calM(\vb{v}, \calM \lor \calU^k))] \\
    &\geq \quad (1/e + \varepsilon/2) \E_{\vb{u} \sim \mu}[w(\OPT(\vb{u}, \calU_m))] \enspace ,
\end{align*}
which is a contradiction by~\Cref{thm:HardDistribution}. Thus, no algorithm is $(1/e+\varepsilon)$-competitive for $\calM \lor \calU^k$ over $\nu$. \Cref{thm:DistributionallyCompetitive} completes the proof.
\end{proof}

\bibliographystyle{alpha}
\bibliography{MasterBib}

\end{document}